\documentclass[conference,a4paper]{IEEEtran}

\usepackage{graphicx}
\usepackage{mathtools}
\usepackage{etoolbox}
\usepackage{lipsum}
\usepackage{amsfonts}
\usepackage{cite}
\usepackage{amsmath,amsthm}
\usepackage{amssymb}

\let\bbordermatrix\bordermatrix
\patchcmd{\bbordermatrix}{8.75}{4.75}{}{}
\patchcmd{\bbordermatrix}{\left(}{\left[}{}{}
\patchcmd{\bbordermatrix}{\right)}{\right]}{}{}

\newtheorem{theorem}{Theorem}[section]
\newtheorem{lemma}[theorem]{Lemma}
\newtheorem{definition}[theorem]{Definition}

\newtheorem{proposition}[theorem]{Proposition}

\newcommand{\sr}{\stackrel}

\newcommand{\rar}{\rightarrow}

\newcommand{\tri}{\sr{\triangle}{=}}

\newcommand{\be}{\begin{equation}}
\newcommand{\ee}{\end{equation}}
\newcommand{\bea}{\begin{eqnarray}}
\newcommand{\eea}{\end{eqnarray}}
\newcommand{\bes}{\begin{eqnarray*}}
\newcommand{\ees}{\end{eqnarray*}}
\newcommand{\bi}{\begin{itemize}}
\newcommand{\ei}{\end{itemize}}
\newcommand{\ben}{\begin{enumerate}}
\newcommand{\een}{\end{enumerate}}
\newcommand{\bp}{\begin{problem}}
\newcommand{\ep}{\end{problem}}

\newcommand{\hst}{\hspace{.2in}}

\newcommand{\noi}{\noindent}

\begin{document}

\title{Nonanticipative Rate Distortion Function\\ for General  Source-Channel Matching}

\author{
  \IEEEauthorblockN{Christos Kourtellaris, Charalambos D.~Charalambous, Photios A.~Stavrou}
  \IEEEauthorblockA{Dep. of Electrical \& Computer Engineering, University of Cyprus, Nicosia, Cyprus\\
    Email: \{kourtellaris.christos, chadcha, stavrou.fotios\}@ucy.ac.cy}

}



\maketitle

\begin{abstract}
In this paper we invoke a nonanticipative information Rate Distortion Function (RDF) for sources with memory, and we analyze its importance in probabilistic matching of the source to the channel so that transmission of a symbol-by-symbol code with memory without anticipation
 is optimal, with respect to an average distortion and excess distortion
probability. We show achievability of the symbol-by-symbol code with memory without anticipation, and we evaluate  the probabilistic performance of the  code for a Markov source.
\end{abstract}

\section{Introduction}

\par We consider a nonanticipative information Rate Distortion Function (RDF) for
sources with memory, and we investigate its importance in joint source-channel coding  JSCC with emphasis on symbol-by-symbol code with memory without anticipation  (e.g. the encoder
and decoder at each time $i$ process samples independently, with memory on past symbols, and without anticipation with respect to symbols occurring at times $j>i$). The aim is to
match probabilistically the source to the channel, and evaluate its performance with respect to average
distortion and excess distortion probability. For memoryless sources and channels, necessary
and sufficient conditions for symbol-by-symbol transmission are given in \cite{gastpar2003} (see also \cite{kverdu})

\par In this paper, we first observe that a necessary condition for probabilistic matching of a source with memory to the channel
so that symbol-by-symbol transmission with memory without anticipation is feasible, is the realization of the optimal reproduction distribution by a cascade of
an encoder-channel-decoder processing information causally. Consequently, we consider a nonanticipative information RDF which
is realizable in the above sense, and we proceed to obtain the closed form expression of  the
reproduction distribution which achieves the infimum over the fidelity set.
Moreover, we prove under certain conditions involving the nonanticipative information RDF, and the capacity
of certain channels with memory and feedback, that symbol-by-symbol code with memory without anticipation is achievable.
\par Finally we evaluate the performance of a stationary ergodic Markov source using
symbol-by-symbol uncoded  transmission  (e.g., the encoder and decoder are unitary operations to their inputs), with
 the channel replaced by the optimal reproduction conditional distribution of the nonanticipative RDF (e.g., the source is not  matched to the channel), by computing an upper bound on the excess distortion probability using a variation of Hoeffding's inequality \cite{glynn2002}.
 Finally we note that nonanticipative information RDF is investigated by the authors in the context of realizable filters in \cite{charalambous-stavrou-ahmed2013}, where examples are given for multi-dimensional partially observable Gaussian processes.

\section{Symbol-by-Symbol codes with Memory Without Anticipation}
\label{sbs}
\par In this section we define the elements of a  symbol-by-symbol code with memory without anticipation.

Let ${\mathbb{N}}\tri\{0,1,\dots\}$, $\mathbb{N}^n\tri\{0,1,\dots,n\}$. The spaces ${\cal X},{\cal A},{\cal B},{\cal Y}$ denote the source output, channel input,
channel output, and decoder output alphabets, respectively, which are assumed to be complete separable metric spaces (Polish spaces)  to avoid excluding continuous alphabets. We define their product spaces by  ${\cal X}_{0,n}\tri\times_{i=0}^{n}{\cal X}$,
${\cal A}_{0,n}\tri\times_{i=0}^{n}{\cal A}$, ${\cal B}_{0,n}\tri\times_{i=0}^{n}{\cal B}$,
${\cal Y}_{0,n}\tri\times_{i=0}^{n}{\cal Y}$. Let
$x^n\tri\{x_0, x_1,\dots, x^n\}\in{\cal X}_{0,n}$ denote the source sequence
of length $n$, and similarly for channel input, channel output, decoder  (reproduction) output sequences,
$a^n\in{\cal A}_{0,n}$, $b^n\in{\cal B}_{0,n}$, $y^n\in{\cal Y}_{0,n}$, respectively. We
associate the above product spaces by their measurable spaces, as usual. Next, we introduce
the various distributions of the blocks appearing in Fig.\ref{cs}.

\begin{definition}
\label{source}
(Source) The source is a sequence of conditional
distributions $\{P_{X_i|X^{i-1}}(d{x_i}|x^{i-1}): \ \forall i\in{\mathbb{N}}^n\}$ defined by
\vspace{-0.3cm}
\begin{align}
P_{X^n}(d{x}^n)\tri\otimes_{i=0}^{n} P_{X_i|X^{i-1}}(d{x}_i|x^{i-1}).\nonumber
\end{align}
\end{definition}
\vspace{-0.2cm}

\begin{figure}
\begin{center}
\includegraphics[bb= -10 23 400 200, scale=0.5]{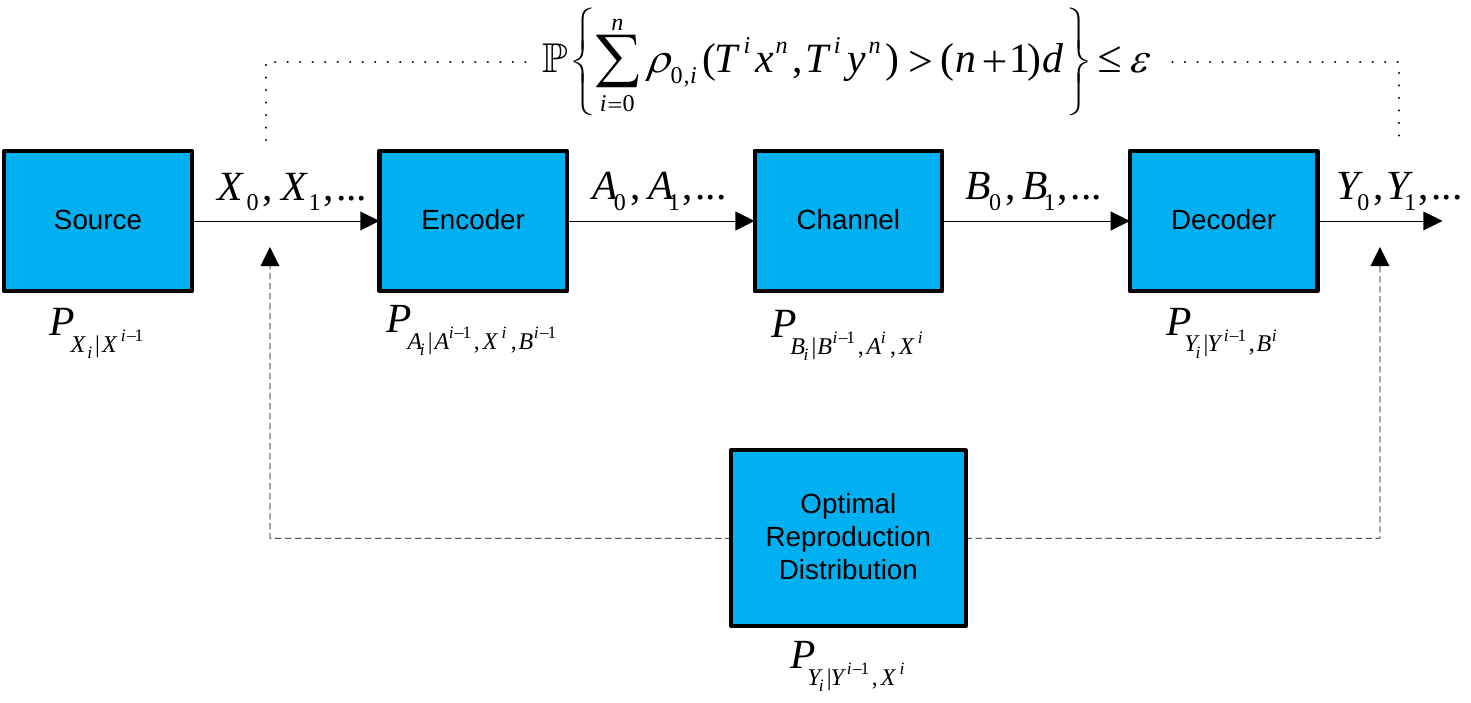}
\caption{Communication scheme with feedback.}
 \label{cs}
\end{center}
\end{figure}

\begin{definition}
\label{encoder}
(Encoder) The encoder is a sequence of conditional
distributions $\{P_{A_i|A^{i-1},B^{i-1},X^i}(d{a_i}|a^{i-1},b^{i-1},x^i):\ \forall i\in{\mathbb{N}}^n\}$
defined by
\begin{align}
&{\overrightarrow P}_{A^n|B^{n-1},X^n}(d{a}^n|b^{n-1},x^n)\nonumber\\
&\tri\otimes_{i=0}^{n}P_{A_i|A^{i-1},B^{i-1},X^i}(d{a}_i|a^{i-1},b^{i-1},x^i). \nonumber
\end{align}
\end{definition}
\vspace{-0.1cm}
Thus, the encoder is nonanticipative in the sense that at each time $i\in{\mathbb{N}}^n$, $P_{A_i|A^{i-1},B^{i-1},X^i}(d{a}_i|a^{i-1},b^{i-1},x^i)$ is a measurable function of past and present symbols $x^i\in{\cal X}_{0,i}$ and past symbols $a^{i-1}\in{\cal A}_{0,i-1}, b^{i-1}\in{\cal B}_{0,i-1}$.
\begin{definition}
\label{channel}
(Channel) The channel is a sequence of conditional
distributions $\{P_{B_i|B^{i-1},A^{i},X^i}(d{b}_i|b^{i-1},a^{i},x^i):\ \forall i\in{\mathbb{N}}^n\}$
defined by
\begin{align}
&{\overrightarrow P}_{B^n|A^{n},X^n}(d{b}^n|a^{n},x^n)\nonumber\\
&\tri\otimes_{i=0}^{n}P_{B_i|B^{i-1},A^{i},X^i}(d{b}_i|b^{i-1},a^{i},x^i).\nonumber
\end{align}
\end{definition}

Thus the channel has memory, feedback and it is nonanticipative with respect to  the source sequence.

\begin{definition}
\label{decoder}
(Decoder) The decoder is a sequence of conditional
distributions $\{P_{Y_i|Y^{i-1},B^{i}}(d{y}_i|y^{i-1},b^{i}): \forall i\in{\mathbb{N}}^n\}$
\begin{align}
{\overrightarrow P}_{Y^n|B^{n}}(d{y}^n|b^{n})\tri
\otimes_{i=0}^{n}P_{Y_i|Y^{i-1},B^i}(d{y}_i|y^{i-1},b^i).\nonumber
\end{align}
\end{definition}

Definitions~\ref{source}-\ref{decoder}, of source-encoder-channel-decoder are  general,  they have memory and feedback  without anticipation, hence we call the source-channel code symbol-by-symbol code with memory without anticipation.
Given the source, encoder, channel,  decoder, we can define uniquely
the joint measure  by
\begin{align}
& P_{X^n,A^n,B^n,Y^n}(d{x}^n,d{a}^n,d{b}^n,d{y}^n)\nonumber\\
& =\otimes_{i=0}^{n}P_{Y_i|Y^{i-1},B^i}(d{y}_i|y^{i-1},b^i)\nonumber\\
&\otimes P_{B_i|B^{i-1},A^{i},X^i}(d{b}_i|b^{i-1},a^{i},x^i)\nonumber\\
&  \otimes P_{A_i|A^{i-1},B^{i-1},X^i}(d{a}_i|a^{i-1},b^{i-1},x^i)\otimes P_{X_i|X^{i-1}}(d{x}_i|x^{i-1}). \label{joint}
\end{align}
Thus, we have indirectly assumed  the following Markov chains (MCs) hold.
\begin{align}
& (A^{i-1},B^{i-1},Y^{i-1})\leftrightarrow X^{i-1}\leftrightarrow X_i, \ \ \forall i\in\mathbb{N}^n \label{mc1} \\
& Y^{i-1}\leftrightarrow (A^{i-1},B^{i-1},X^{i})\leftrightarrow A_i, \ \ \forall i\in\mathbb{N}^n \label{mc2}\\
& Y^{i-1}\leftrightarrow (A^{i},B^{i-1},X^{i})\leftrightarrow B_i, \ \ \forall i\in\mathbb{N}^n \label{mc3}\\
& (A^{i},X^{i})\leftrightarrow (B^i, Y^{i-1})\leftrightarrow Y_i, \ \ \forall i\in\mathbb{N}^n. \label{mc4}
\end{align}
\noi The distortion function between the source and its reproduction is a measurable function $d_{0,n}:{\cal X}_{0,n}\times{\cal Y}_{0,n}\mapsto [0,\infty)$,
\begin{align}
d_{0,n}(x^n,y^n)\tri\sum_{i=0}^{n}{\rho}_{0,i}({T}^i{x^n},T^i{y^n})\nonumber
\end{align}
where $({T}^i{x^n},{T}^i{y^n})$ are the shift operations on $(x^n,y^n)$, respectively.
For a single letter distortion function we take ${\rho}_{0,i}(T^ix^n,T^iy^n)={\rho}(x_i,y_i)$.
The cost of transmitting  symbols over
the channel is a measurable function
\begin{align}
&c_{0,n}\hspace{-0.1cm}:\hspace{-0.1cm}{\cal A}_{0,n}\hspace{-0.1cm}\times\hspace{-0.1cm}{\cal Y}_{0,n-1}\hspace{-0.1cm}\mapsto \hspace{-0.1cm}[0,\infty), c_{0,n}(a^n,y^{n-1})\hspace{-0.1cm}\tri\hspace{-0.1cm}\sum_{i=0}^{n}
{\gamma}_{0,i}(a^i,y^{i-1})\nonumber
\end{align}
\noi Next, we state the definition of a symbol-by-symbol code with memory without anticipation.

\begin{definition}
\label{sbsc}
(Symbol-by-Symbol code with Memory without Anticipation).
 An (n,d,$\epsilon$,P) symbol-by-symbol
code with memory without anticipation for (${\cal X}_{0,n}, {\cal A}_{0,n}, {\cal B}_{0,n}, {\cal Y}_{0,n}, P_{X^n},
{\overrightarrow P}_{B^n|A^n,X^n}, d_{0,n}, c_{0,n}$) is a  code
$\{P_{A_i|A^{i-1}, B^{i-1},X^i}(\cdot|\cdot):\forall i \in\mathbb{N}^n\}$, $\{P_{Y_i|Y^{i-1},B^i}(\cdot|\cdot):\forall i \in\mathbb{N}^n\}$
with excess distortion probability
\begin{align}
{\mathbb P}\Big\{d_{0,n}(x^n,y^n)>(n+1)d\Big\}\leq\epsilon, \ \epsilon\in(0,1), \ d\geq 0 \nonumber
\end{align}
and transmission cost $\frac{1}{n+1} {\mathbb E}\Big\{c_{0,n}(A^n,Y^{n-1})\Big\}\leq P, \  P\geq 0$.
\end{definition}

\begin{definition}(Minimum Excess Distortion)
\label{asbsc}
The minimum excess distortion achievable by a symbol-by-symbol
code with memory without anticipation $(n,d,\epsilon,P)$ is defined by
\begin{align}
&  D^o(n,\epsilon, P)\tri\inf\Big\{d:  \exists(n,d,\epsilon,P)\  \ \mbox{symbol-by- symbol code} \nonumber \\
&  \mbox{with memory without anticipation}\Big\} \nonumber
\end{align}
\end{definition}

Our definition of symbol-by-symbol code with memory without anticipation is randomized, hence it embeds  deterministic codes as a special case
  \cite{kverdu}.

\section{Nonanticipative Versus Classical RDF}
\label{cnrdf}
In this section, we first establish the claim that the classical RDF for sources with memory, is not
the appropriate measure for lossy compression in symbol-by-symbol codes with memory without anticipation. Recall that the necessary conditions for transmission of
symbol-by-symbol codes with memory without anticipation (this is also true for  memoryless sources and channels) are the following.

\begin{enumerate}
\item Realization
of the optimal reproduction distribution of lossy compression with fidelity by an encoder-channel-decoder scheme, processing information causally;

\item Computation of the RDF and that of the optimal reproduction distribution so that  probabilistic matching of the source and  channel is feasible.
\end{enumerate}
%

\par Consider the average fidelity set
\begin{align}
&{\cal Q}_{0,n}(D)\tri \Big\{{P}_{Y^n|X^n}: \nonumber  \\
&\frac{1}{n+1}\int d_{0,n}(x^n,y^n)
({P}_{Y^n|X^n}\otimes P_{X^n})(dx^n,dy^n)\leq D\Big\}.  \nonumber 
\end{align}
 Here, $P_{X^n}(\cdot)$ is the source distribution and
$P_{Y^n|X^n}(\cdot|x^n)$ is the reproduction distribution, and
it is known that for stationary ergodic source and single letter distortion, the OPTA   is given by the RDF \cite{berger}
\begin{align}
R(D)&=\lim_{n\rar \infty} R_{0,n}(D) \label{rdf} \\
R_{0,n}(D)&=\inf_{{P}_{Y^n|X^n}\in {\cal Q}_{0,n}(D)}\frac{1}{n+1}I(X^n;Y^n) \label{rdfon3}
\end{align}
provided the infimum is achievable. It is also well known that if the infimum in (\ref{rdfon3}) exists, then \cite{berger}
\begin{align}
{P}_{Y^n|X^n}(dy^n|x^n)=\frac{e^{sd_{0,n}({x}^n, { y}^{n})}P_{Y^n}(d{y}^n)}
{\int_{{\cal  Y}_{0,n}}e^{sd_{0,n}({x}^n, { y}^{n})}P_{Y^n}(d{y}^n)} \label{oprecdnc}
\end{align}
where $s\in(-\infty,0]$ is the Lagrange multiplier associated with the fidelity set ${\cal Q}_{0,n}(D)$. Clearly, by Bayes' rule
\begin{align}
{P}_{Y^n|X^n}(dy^n|x^n) =&\otimes_{i=0}^n {P}_{Y_i|X^n, Y^{i-1}}(dy_i|x^n,y^{i-1})  \label{cr}
\end{align}
and hence the optimal reproduction $y_i$ at time $i$ of $x_i$ depends on the past reproductions and past and present source symbols $\{y^{i-1}, x^i\}$, and the future source symbols $\{x_{i+1}, \ldots x^n\}, n \geq i$. Thus, in general
the optimal reproduction distribution is anticipative with respect to the source symbols, and hence it is not realizable in the sense described earlier.
 Moreover for sources with memory it is very difficult to compute the  value of $R(D)$. Even for the Binary Symmetric Markov Source (BSMS) the exact expression of $R(D)$ is not known \cite{berger1977}.  The independent source and Gaussian source are exception.

Now, we  introduce the nonanticipative information RDF which by construction is realizable, and  in Section \ref{ord} we compute its closed form expression.
 Given a source ${P}_{X^n}(dx^n)$ and a causal conditional distribution defined by
\begin{align}
\overrightarrow{P}_{Y^n|X^n}(dy^n|x^n) \tri \otimes_{i=0}^{n}P_{Y_i|Y^{i-1},X^i}(dy_i|y^{i-1},x^i) \label{cc1}
\end{align}
then the joint distribution $P_{Y^n,X^n}$ and marginal distribution $P_{Y^n}$ are uniquely defined.
 Introduce the information measure ($\mathbb{D}(.|.)$ denotes the relative entropy).
\begin{align}
I_{P_{X^n}}(X^n\rar Y^n)&\tri\mathbb{D}({\overrightarrow P}_{Y^n|X^n}\otimes P_{X^n}||{P}_{Y^n}\times P_{X^n})\nonumber\\
&\equiv\mathbb{I}_{X^n\rightarrow{Y^n}}(P_{X^n},{\overrightarrow P}_{Y^n|X^n}).\nonumber
\end{align}
Consider the fidelity set defined by
\begin{align}
\overrightarrow{\cal Q}_{0,n}(D) &\tri  \Big\{{\overrightarrow P}_{Y^n|X^n}: \frac{1}{n+1}\int_{{\cal X}_{0,n}\times {\cal Y}_{0,n}}d_{0,n}({x^n}
,{y^n})\nonumber\\
&\qquad {\overrightarrow P}_{Y^n|X^n}(dy^n|x^n)\otimes P_{X^n}(dx^n) \leq D\Big\}.\label{nafs}
\end{align}
Next, we introduce the nonanticipative information  RDF.

\begin{definition}(Nonanticipative Information RDF)
\label{nardf}
Given ${\overrightarrow {\cal Q}}_{0,n}(D)$, the nonanticipative information RDF is defined by
\begin{align}
{ R}^{na}_{0,n}(D) \tri \inf_{ \overrightarrow{P}_{Y^n|X^n}\in   {\overrightarrow {\cal Q}}_{0,n}(D)} \frac{1}{n+1} \mathbb{I}_{X^n\rightarrow{Y^n}}(P_{X^n},{\overrightarrow P}_{Y^n|X^n})\label{nardf11}
\end{align}
and its rate by ${ R}^{na}(D)=\lim_{n\rar\infty}{ R}^{na}_{0,n}(D)$
provided infimum and the limit exist.
\end{definition}

Clearly, if the minimum of ${ R}^{na}_{0,n}(D)$ exists the  optimal reproduction distribution
is nonanticipative, and hence realizable in the sense described before.

Next, draw the connection  between $R_{0,n}^{na}(D)$, ${ R}_{0,n}(D)$, and Gorbunov-Pinsker definition of nonanticipatory $\epsilon-$entropy, by first introducing the following equivalent statements of conditional independence.

\begin{lemma}\label{equivsta}
(Equivalent Statements of Nonanticipation)
The following are equivalent for $i=0,1\ldots,n-1$, $\forall n\in{\mathbb{N}}$.
\begin{enumerate}
\item $X_{i+1}^n\leftrightarrow (X^i,Y^{i-1}) \leftrightarrow Y_i $ forms a MC;
\item $X_{i+1}\leftrightarrow X^i \leftrightarrow Y^i $ forms a MC;
\item $X_{i+1}^n\leftrightarrow X^i \leftrightarrow Y^i $ forms a MC;
\item $P_{Y^n|X^n}(dy^n|x^n)= \overrightarrow{P}_{Y^n|X^n}(dy^n|x^n).$
\end{enumerate}
\end{lemma}
\begin{proof}
The equivalency of 1),  2), 4) is easy. If 3) holds then $P_{X_{i+1}^n|X^i,Y^i}
(d{x}_{i+1}^n|x^i,  y^i)= P_{X_{i+1}^n|X^i}(d{x}_{i+1}^{n}|x^i)$ and hence 2) is obtained by integration.
By induction one can show that 2) implies 3).
\end{proof}

Clearly, ${ R}^{na}_{0,n}(D) \geq  R_{0,n}(D)$.
Next, we discuss the relation between ${R}^{na}_{0,n}(D)$ and  Gorbunov and Pinsker \cite{pinsker1973} nonanticipatory  $\epsilon-$entropy. Gorbunov and Pinsker \cite{pinsker1973}, restricted the fidelity set ${\cal Q}_{0,n}(D)$
to those reproduction distributions which satisfy the MC of Lemma~\ref{equivsta}, 3),
and introduced the nonanticipatory $\epsilon$-entropy defined by \cite{pinsker1973}
\begin{align}
R^{\varepsilon}_{0,n}(D)=\mathop{\inf_{{P}_{Y^n|X^n}\in {\cal Q}_{0,n}(D)}}_{{X_{i+1}^{n}}\leftrightarrow X^i \leftrightarrow Y^i, \ i=0,1,\dots,n-1
}\frac{1}{n+1}I(X^n;Y^n) \label{rdfon}
\end{align}
and the nonanticipatory message generation of the source by
$R^{\varepsilon}(D)=\lim_{n\rar\infty}R^{\varepsilon}_{0,n}(D)$
provided the infimum exists and the limit is finite. The MC in (\ref{rdfon}) means that  the reproduction distribution which minimizes
(\ref{rdfon}) can be realized via an encoder-channel-decoder, using nonanticipative operations (causal).\\
\noi In view of Lemma~\ref{equivsta} we have the following theorem.
\begin{theorem}
\label{gp}
(Equivalent Nonanticipative RDF)
The following holds
\begin{align}
R^{\varepsilon}_{0,n}(D)=R^{na}_{0,n}(D),  \hst \forall n\in\mathbb{N}. \label{gpe}
\end{align}
\end{theorem}

\begin{proof}
If  any of the statements of Lemma~\ref{equivsta} hold then  $I(X^n;Y^n)=I(P_{x^{n}},{\overrightarrow P}_{Y^n|X^n})$, and the fidelity set is (\ref{nafs}).
\end{proof}

\section{Solution  Nonanticipative RDF}
\label{ord}
In this section we give the expression of the nonanticipative reproduction distribution which  achieves the
infimum in (\ref{nardf11}). First, we  note that in view of Theorem~\ref{gp}, the results derived in \cite{pinsker1973} are applicable for $R_{0,n}^{na}(D), R^{na}(D)$, and these results include sufficient conditions for stationary sources to give an optimal reproduction distribution corresponding to stationary source-reproduction pair $\{(X_i,Y_i): i=0,1,\ldots\}$.

Thus, under the conditions in  \cite{pinsker1973} or assuming the solution of $R_{0,n}^{na}(D)$ gives an optimal nonanticipative reproduction distribution which is stationary, and hence ${\overrightarrow P}_{Y^n|X^n}(d{y}^n|x^n)$ is an $(n+1)-$fold  convolution
of stationary conditional distributions, we have the following theorem.

\begin{theorem}
\label{maintheo}
Suppose there exist an interior point of the fidelity set, and the optimal reproduction is stationary.  Then the infimum over $\overrightarrow{\cal Q}_{0,n}(D)$ in (\ref{nardf11}) is
attained by
\begin{align}
\overrightarrow{P}_{Y^n|X^n}^*(dy^n|x^n)=\otimes_{i=0}^{n}\frac{e^{s\rho(T^i{x}^n,T^i{ y}^{n})}
P_{Y_i|Y^{i-1}}^*(d{y}_{i}|{ y}^{i-1})}{\int_{{\cal  Y}_i}
e^{s\rho(T^i{x}^n,T^i{ y}^{n})}P_{Y_i|Y^{i-1}}^*(d{ y}_{i}|{ y}^{i-1})}  \label{crdoo}
\end{align}
where $s\leq 0$ is the Lagrange multiplier associated with the constraint which is satisfied with equality, and
\begin{align}
{ R}^{na}_{0,n}(D)=& sD - \frac{1}{n+1} \sum_{i=0}^{n}\int_{{\cal X}_{0,i}\times{{\cal Y}}_{0,i-1}}
\log\Big(\int_{{{\cal Y}}_{i}}e^{s\rho(T^i{x}^n, T^i{ y}^{n})}\nonumber\\
& P_{Y_i|Y^{i-1}}^{*}(d{y}_{i}|{ y}^{i-1})\Big) \otimes {P}_{X_i|X^{i-1}}(d{x}_i|{ x}^{i-1})\nonumber \\
&\otimes P_{X^{i-1},Y^{i-1}}^*(d{x}^{i-1},d{y}^{i-1}){\label{solrdf}}
\end{align}
where $P_{X^{i-1},Y^{i-1}}^*(\cdot,\cdot)= \overrightarrow{P}_{Y^{i-1}|X^{i-1}}^*(\cdot|\cdot)\otimes P_{X^{i-1}}(\cdot)$.
\end{theorem}
\begin{proof}  The derivation is given in \cite{charalambous-stavrou-ahmed2013}.
\end{proof}

The point to be made regarding the optimal reproduction distribution is that, it is nonanticipative, and  as we show in the next section, easy to compute, even for sources with memory.

\section{Coding Theorem}
\label{coding}
\par In this section we show achievability of symbol-by-symbol code with memory without anticipation. We also note
that in view of the equivalence ${ R}^{\varepsilon}_{0,n}(D)={ R}_{0,n}^{na}(D)$, that ${ R}^{na}_{0,n}(D)$ is the OPTA by sequential code (see \cite{tatikonda2000}).

\par The probabilistic realization of the optimal reproduction distribution by an encoder-channel-decoder, is necessary for probabilistic matching of the source
and the channel. Next, we give the precise definition of the realization.

\begin{definition}\label{realdef}
(Realization)
Given a source $\{P_{X_i|X^{i-1}}$ $(d{x}_i|x^{i-1}): \forall i \in {\mathbb N}^n\}$, a general channel
$\{P_{B_i|B^{i-1},A^i,X^i}$ $(d{b}_i|b^{i-1},a^i,x^i): \forall i \in {\mathbb N}^n\}$  is a realization
of the optimal reproduction distribution $\{P_{Y_i|Y^{i-1},X^i}^*(d{y}_i|y^{i-1},x^i): \forall i \in {\mathbb N}^n\}$
of theorem \ref{maintheo}, if there exists a pre-channel encoder
$\{P_{A_i|A^{i-1},B^{i-1},X^i}$ $(d{a}_i|a^{i-1},b^{i-1},x^i): \forall i \in {\mathbb N}^n\}$ and a post-channel
decoder $\{P_{Y_i|Y^{i-1},B^{i}}$ $(d{y}_i|y^{i-1},b^{i}): \forall i \in {\mathbb N}^n\}$ such that
\begin{align}
{\overrightarrow P}^*_{Y^n|X^n}(d{y}^n|x^n)&=\otimes_{i=0}^n{ P}^*_{Y_i|Y^{i-1},X^i}
(d{y}_i|y^{i-1},x^i)\nonumber\\
&=\otimes_{i=0}^n{ P}_{Y_i|Y^{i-1},X^i}
(d{y}_i|y^{i-1},x^i)\label{scmrd}
\end{align}
where the joint distribution from which (\ref{scmrd}) is obtained  is  precisely (\ref{joint}).
Moreover we say that ${ R}^{na}_{0,n}(D)$ is realizable if in addition the realization
operates with average distortion $D$ and $I_{P_{X^n}}(P_{X^n},\overrightarrow{P}_{Y^n|X^n})={ R}^{na}_{0,n}(D)$
\end{definition}

If the optimal reproduction distribution is realizable (see  Definition~\ref{realdef}), then  the data processing inequality holds:
\begin{align}
I_{X^n\rar Y^n}(P_{X^n},{\overrightarrow P}_{Y^n|X^n})\leq I(X^n\rar B^n), \ \forall n \in{\mathbb{N}}. \label{dpi}
\end{align}
If ${ R}^{na}_{0,n}(D)$ is realizable according to Definition \ref{realdef},
then the source is not necessarily matched to the channel. Next, we  prove (under certain conditions) achievability, by first introducing the information
definition of channel capacity. \\
Consider the following average cost set defined by
\begin{align}
{\cal P}_{0,n}(P)\tri\Big\{(X^n,A^n):\frac{1}{n+1}{\mathbb E}\{c_{0,n}(A^n,Y^{n-1})\}\leq P\Big\}. \nonumber
\end{align}
Since we consider the general scenario that (\ref{mc1})-(\ref{mc4}) hold, then we define the information channel capacity from the source to the channel output as  follows \cite{cover-pombra1989}.
\begin{align}
C_{0,n}(P)\tri\sup_{(X^n,A^n)\in{\cal P}_{0,n}(P)}\frac{1}{n+1}I(X^n\rar B^n)\nonumber
\end{align}
and its rate (provided $\sup$ is finite and the limit exists) by
$C(P)=\lim_{n\rar\infty}C_{0,n}(P)$.

Next, we prove achievability of a symbol-by-symbol code.

\begin{theorem}\label{ach}
(Achievability of Symbol-by-Symbol Code with Memory Without Anticipation).\\
Suppose the following conditions hold.

\begin{enumerate}
\item ${ R}^{na}_{0,n}(D)$ has a solution and the optimal reproduction distribution is stationary.

\item $C_{0,n}(P)$ has a solution and the maximizing processes are stationary.

\item The optimal reproduction distribution $\overrightarrow{P}_{Y^n|X^n}(dy^n|x^n)$ given by Theorem \ref{maintheo} is realizable, and ${ R}^{na}_{0,n}(D)$ is also realizable.

\item There exists  $D$ and $P$ such that ${ R}^{na}_{0,n}(D)=C_{0,n}(P)$.
\end{enumerate}
If
\begin{align}
\mathbb{P}\Big\{\sum_{i=0}^{n}{\rho}_{0,i}(T^i{X^n},T^i{Y^n})>(n+1)d\Big\}\leq\epsilon \label{edp}
\end{align}
where ${\mathbb P}$ is taken with respect to $P_{Y^n,X^n}(d{y}^n,d{x}^n)={\overrightarrow P}^*_{Y^n|X^n}(d{y}^n|x^n)\otimes P_{X^n}(d{x}^n)$ then
 there exists an $(n,d,\epsilon,P)$ symbol-by-symbol code with memory without anticipation.
\end{theorem}

\begin{proof}
The derivation is similar to \cite{gastpar2003}. If conditions (1), (3) hold then the optimal reproduction distribution is realizable, and this realization achieves ${ R}^{na}_{0,n}(D)$. By (4) the source is matched to the channel so that the excess distortion probability of a symbol-by-symbol code  with memory without anticipation satisfies (\ref{edp}).
\end{proof}

\subsubsection{\bf Symbol-by-Symbol Code} It can be shown that if the source is Markov, and the channel is Markov with respect to the source, satisfying
\begin{enumerate}
\item $P_{X_i|X^{i-1}}(x_i|x^{i-1})=P_{X_i|X_{i-1}}(x_i|x_{i-1}), \ \forall i\in\mathbb{N}^n$
\item $P_{B_i|B^{i-1},A^{i},X^i}(d{b}_i|b^{i-1},a^{i},x^i) \\=P_{B_i|B^{i-1},A_{i},X_i}(d{b}_i|b^{i-1},a_{i},x_i), \ \forall i\in\mathbb{N}^n$,
\end{enumerate}
then
 maximizing directed information $I(X^{n}\rightarrow  {B}^{n})$  over non-Markov encoders $\{P_{A_i|A^{i-1},B^{i-1},X^i}: i=0,1,\ldots,n\}$ is equivalent to maximizing it over encoders $\{\overline{P}_{A_i|B^{i-1},X_i}: i=0,1,\ldots,n\}$, and similarly,  maximizing $I(X^{n}\rightarrow  {B}^{n})$  over non-Markov deterministic encoders $\{e_i(x^i,a^{i-1},y^{i-1}): i=1, \ldots, n\}$ is equivalent to the maximization with respect to encoders $\{g_i(x_i,y^{i-1}): i=1, \ldots, n\}$. This result appeared in \cite{charalambous-kourtellaris-hadjicostis}. Thus, based on these two conditions the encoder is symbol-by-symbol Markov with respect to the source, and nothing can be gained by considering an encoder that depends on the entire past of the source causally.

\section{Application}\label{exa}
In this section we  consider the Binary Symmetric Markov source, for which the classical RDF is unsolved and only bounds are known. Then we show that the solution of the nonanticipative information  RDF can be obtained relatively easy. Subsequently, we evaluate the performance of uncoded transmission. It is shown that even this uncoded, unmatched scheme, although sub-optimal ensures the excess distortion probability goes to zero.

\par Consider a Binary Symmetric Markov Source (BSMS(p)),
$P(x_i=0|x_{i-1}=0)=P(x_i=1|x_{i-1}=1)=1-p$ and $P(x_i=1|x_{i-1}=0)=P(x_i=0|x_{i-1}=1)=p$ and
$i=0, 1,\dots,n$. We apply a single letter Hamming distortion criterion $\rho(x,y)=0$ if $x=y$ and $\rho(x,y)=1$ if $x \neq y$.  The objective is to compute ${R}^{na}(D)$.

\begin{proposition}\label{marex1} For a BSMS(p) and single letter distortion criterion we have
\[ { R}^{na}(D) = \left\{ \begin{array}{ll}
         H(m)-H(D) & \mbox{if $D \leq \frac{1}{2}$}\\
        0 & \mbox{otherwise}\end{array}  \right. \]
where $m=1-p-D+2pD$.
\end{proposition}
\begin{proof}
We describe the main steps. The steady state
distribution of the source is $P(X_i=0)=P(X_i=1)=0.5$ and the reproduction distribution is
\begin{align}
P_{Y_i|X^i,Y^{i-1}}^*=P_{Y_i|X_i,Y^{i-1}}^*=
\frac{e^{s\rho(x_i,y_i)}P(y_i|y^{i-1})}{\sum_{y_i}e^{s{\rho}(x_i,y_i)}P(y_i|y^{i-1})}\nonumber
\end{align}
and we can show that  $P_{Y_i|X_i,Y^{i-1}}^*=P_{Y_i|X_i,Y_{i-1}}^*$ and that
\begin{align}
P_{Y_i|X_i,Y_{i-1}}^*(y_i|x_i,y_{i-1})=\bbordermatrix{~ & 0,0 & 0,1 & 1,0 & 1,1 \cr
                  0 & \alpha & \beta& 1-\beta & 1-\alpha\vspace{0.3cm} \cr
                  1 & 1-\alpha & 1-\beta& \beta &  \alpha \cr}\nonumber
\end{align}
where $\alpha=\frac{(1-p)(1-D)}{1-p-D+2pD}$,  $\beta=\frac{p(1-D)}{p+D-2pD}$.

\begin{figure}
\begin{center}
\includegraphics[bb= -10 33 400 310,scale=0.52]{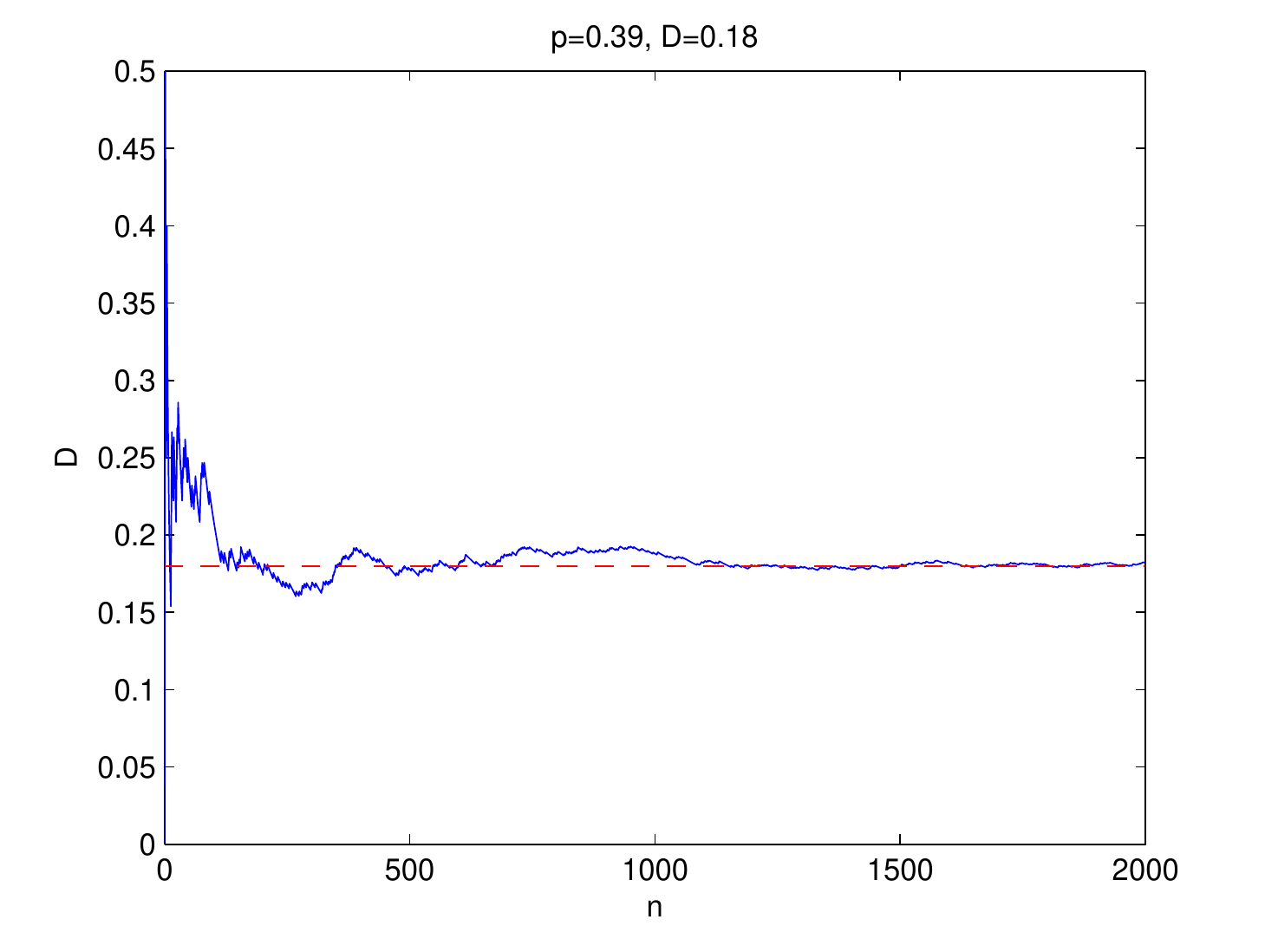}
\caption{The distortion between the source and reproduction symbols for a random realization of the source, as a function of $n$
using the optimal reproduction distribution as the channel and uncoded transmission.}
 \label{figmarex1}
\end{center}
\end{figure}
\end{proof}

\begin{figure}
\begin{center}
\centering
\includegraphics[bb= -10 27 400 310,scale=0.52]{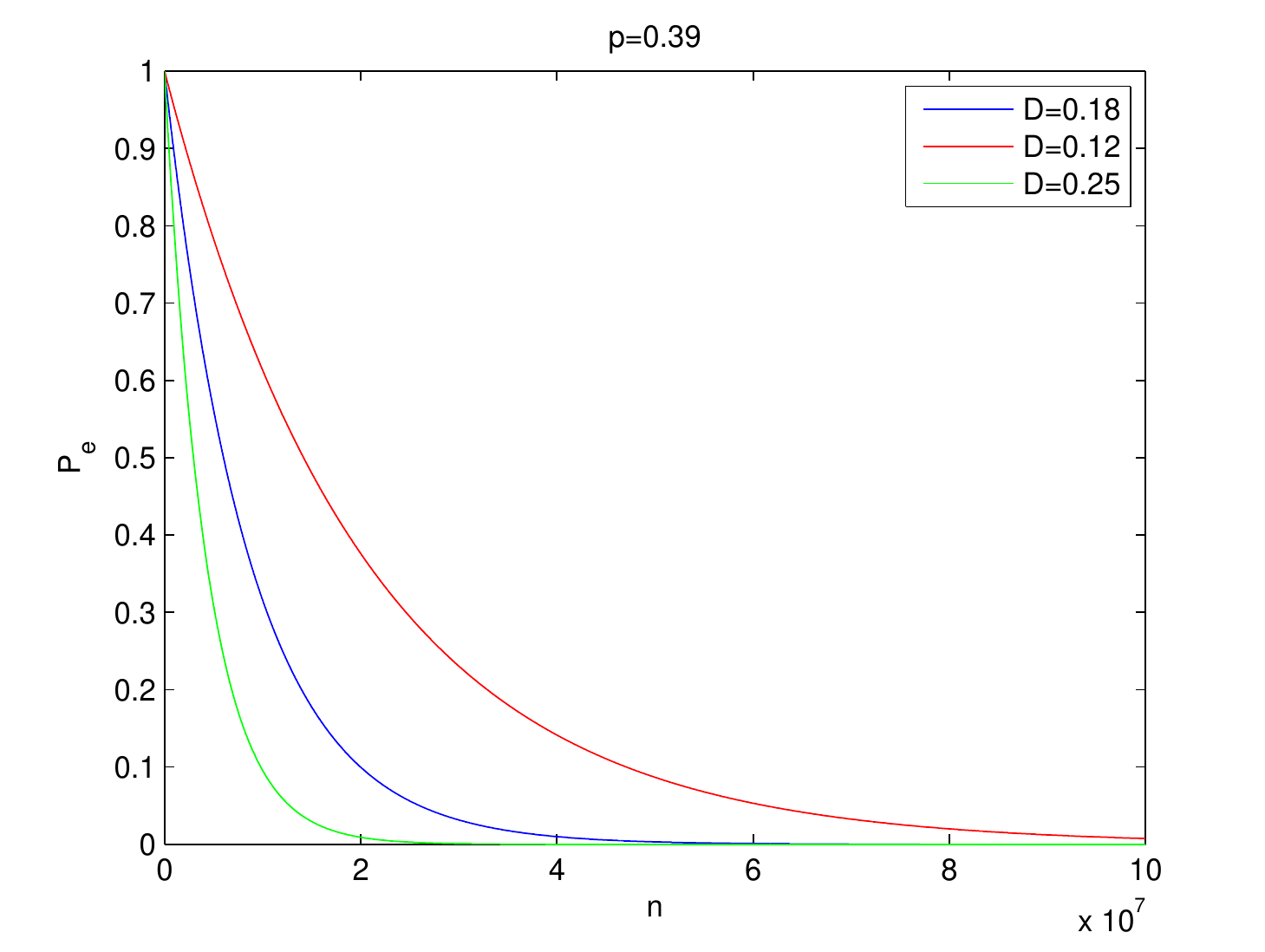}
\caption{Excess Probability of Distortion for $\delta=0.01$.}
 \label{figmarex2}
\end{center}
\end{figure}

\par Next, we discuss symbol-by-symbol uncoded transmission
over a channel characterized via the optimal reproduction distribution. This approach is suboptimal
since the channels capacity is not necessarily matched to the source RDF. The matching is part of on-going research and it could be possible by adding a cost
constrain on the channel. A  realization of the described scheme is shown in Fig.~\ref{figmarex1} , where it is verified that
as the number of channel uses $n$ is increased, the single letter distortion between the source symbol sequence and
the reproduction sequence converges to the average distortion $D$.

\par Next, we bound the excess distortion probability of Theorem \ref{ach}, by applying an extension of Hoeffding's inequality for MCs \cite{glynn2002}, which
bounds the probability of a function of a Markov source. It can be shown that  $\{Z_i \tri (Y_i,X_i): \forall i \in {\mathbb N}  \}$ is Markov. Set ${\rho}(x,y)=x\oplus y$ and let $S_n \tri \sum_{i=0}^{n}{\rho}(X_i,Y_i)$.
Let $d\tri \delta+ \frac{{\mathbb E}[S_n]}{n+1}, \delta >0$.
By Hoeffding's inequality, the excess distortion probability  is bounded by
\bea
P\Big\{S_n   > (n+1) d \Big\}\leq \exp\Big(-\frac{{\lambda}^2
((n+1)\delta -2\|f\|m/{\lambda})^2}{2(n+1){\|f\|}^2m^2}\Big)\nonumber
\eea
where ${\|f\|}=1$, $m=1$,
$\lambda=\min\{p,1-p\}\min\{\alpha,\beta,1-\alpha,1-\beta\}$, for $n>2{\|f\|}m/(\lambda\delta)$.
This bound is illustrated in Fig.~\ref{figmarex2}.  Although, this bound is not tight and holds for
$n$ large enough, it shows the achievability of
Markov sources via uncoded transmission. It might be possible to compute the excess distortion probability in closed form to get tighter bounds.

\section{Conclusions}
This paper considers nonanticipative information RDF and discusses its application to General Source-Channel Matching, generalizing earlier results on uncoded transmission to random processes with memory and nonanticipative feedback.


%



\label{Bibliography}
\bibliographystyle{IEEEtran}
\bibliography{Bibliography}
%
%
%
%
%
%
%

\end{document}